\documentclass[letterpaper, 10pt, conference]{ieeeconf}      % Use this line for a4 paper

\IEEEoverridecommandlockouts                              % This command is only needed if 
\usepackage{mathptmx} % assumes new font selection scheme installed
\usepackage{times} % assumes new font selection scheme installed
\usepackage{amsmath} % assumes amsmath package installed
\usepackage{amssymb}  % assumes amsmath package installed
\usepackage{algorithm}
\usepackage{algorithmic}
\usepackage{textcomp}

\usepackage{graphicx}
\usepackage{epstopdf}
\usepackage{mathrsfs}

\newcommand{\IM}{\mathrm{im}}

\newcommand{\QNUM}{D}
\newcommand{\Rank}{\mathrm{rank}}
\newcommand{\SPAN}{\mathrm{span}}
\newtheorem{Definition}{Definition}

\newtheorem{Lemma}{Lemma}
\newtheorem{Claim}{Claim}

\newtheorem{Corollary}{Corollary}
\newtheorem{Remark}{Remark}
\newtheorem{Notation}{Notation}
\newtheorem{Example}{Example}

\newtheorem{Problem}{Problem}

\newcommand{\LSS}{DTLSS\ }
\newcommand{\BLSS}{DTLSS\ }
\newcommand{\SLSS}{DTLSSs\ }

\newcommand{\GCR}{\textbf{GCR}}

\title{\LARGE \bf
   Model Reduction of Linear Switched Systems by Restricting Discrete Dynamics
}

\author{Mert Ba\c{s}tu\u{g}$^{1,2}$, Mih\'{a}ly Petreczky$^{2}$, Rafael Wisniewski$^{1}$ and John Leth$^{1}$% <-this % stops a space
\thanks{$^{1}$Department of Electronic Systems, Automation and Control, Aalborg University, 9220 Aalborg, Denmark {\tt\small mertb@es.aau.dk, raf@es.aau.dk, jjl@es.aau.dk}}%
\thanks{$^{2}$Department of Computer Science and Automatic Control (UR Informatique et Automatique), \'{E}cole des Mines de Douai, 59508 Douai, France {\tt\small mihaly.petreczky@mines-douai.fr}}%
}

\begin{document}

\maketitle
\thispagestyle{empty}
\pagestyle{empty}

%%%%%%%%%%%%%%%%%%%%%%%%%%%%%%%%%%%%%%%%%%%%%%%%%%%%%%%%%%%%%%%%%%%%%%%%%%%%%%%%
\begin{abstract}

We present a procedure for reducing the number of continuous states of
discrete-time linear switched systems, such that the reduced system has the same
behavior as the original system for a subset of switching sequences. The proposed
method is expected to be useful for abstraction based control synthesis methods
for hybrid systems.

\end{abstract}

%%%%%%%%%%%%%%%%%%%%%%%%%%%%%%%%%%%%%%%%%%%%%%%%%%%%%%%%%%%%%%%%%%%%%%%%%%%%%%%%
\section{INTRODUCTION}

A discrete-time linear switched system \cite{liberzon2003,Sun:Book} (abbreviated by DTLSS) is a discrete-time hybrid system 
%which consists of a finite number of linear subsystems, i.e. 
%it is a system 
of the form
\begin{equation}
\label{dtlss1}
\Sigma\left\{
\begin{split}
  & x(t+1) = A_{\sigma(t)}x(t)+B_{\sigma(t)}u(t) \mbox{ and } x(0)=x_0\\
  & y(t) = C_{\sigma(t)}x(t),
\end{split}\right.
\end{equation}
where $x(t) \in \mathbb{R}^n$ is the continuous state, $y(t) \in \mathbb{R}^p$ the continuous output, 
$u(t) \in \mathbb{R}^{m}$ is the continuous input, $\sigma(t) \in Q=\{1,\ldots,\QNUM\}$, $\QNUM > 0$ is the discrete state (switching signal).
$A_q,B_q,C_q$ are matrices of suitable dimension for $q \in Q$. A more rigorous definition of \SLSS\ will be presented later on.
For the purposes of this paper, $u(t)$ and $\sigma(t)$ will be viewed as externally generated signals.
 %Linear  switched systems have been studied extensively, see \cite{liberzon2003}, \cite{Sun:Book} for an overview. 
%In this paper we will address the problem of reducing the number of continous state variables of a  linear switched system, while preserving the
%input-output behavior of the system for certain switching sequences.

 \textbf{Contribution of the paper}
 Consider a discrete-time linear switched system $\Sigma$ of the form \eqref{dtlss1}, and a
 set $L$ which describes the admissible set of switching sequences. In this paper, we will present an
 algorithm for computing another \LSS
 \begin{equation}
 \label{dtlss2}
  \bar{\Sigma}\left\{
   \begin{split}
   &  \bar{x}(t+1) = \bar{A}_{\sigma(t)}\bar{x}(t)+\bar{B}_{\sigma(t)}u(t) \mbox{ and } \bar{x}(0)=\bar{x}_0\\
   &   \bar{y}(t) = \bar{C}_{\sigma(t)}\bar{x}(t) 
  \end{split}\right.
 \end{equation}
 such that for any switching sequence $\sigma(0)\cdots \sigma(t) = q_0 \cdots q_t \in L$ and continuous inputs $u(0),\ldots, u(t-1)$, 
 the output at time $t$ of \eqref{dtlss1} equals the output of 
 \eqref{dtlss2}, i.e., $y(t)=\bar{y}(t)$ and the number of state variables of \eqref{dtlss2} is smaller than that of \eqref{dtlss1}.
 In short, for any sequence of discrete states from $L$, the  input-output behavior of $\Sigma$ and $\bar{\Sigma}$ coincide and the
 size of $\bar{\Sigma}$ is smaller. 
 %That is, the contribution of the paper can be viewed as a model reduction algorithm, which removes those continous state which do not
 %contribute to the behavior for sequences of discrete modes from $L$. 

\textbf{Motivation}
 %The motivation for considering model reduction of linear switched systems is the same as for any other system  class:
 Realistic plant models of industrial interest tend to be quite large and in general, the smaller is the plant model,
 the smaller is the resulting controller and the computational complexity of the control synthesis or verification algorithm. This is especially
 apparent for hybrid systems, since in this case, the computational complexity of control or verification algorithms is
 often exponential in the number of continuous states \cite{TabuadaBook}.
 %In particular, this is of the case for techninques which are based on computing
 %finite-state abstraction.This means that even plant models of moderate size can become intractable and even a small reduction in the number of states
 %can make a  difference.
 The particular model reduction problem formulated in this paper was motivated by the observation that in many instances, we are interested in the 
 behavior of the model only for certain switching sequences. 
 %Restrictions on switching sequences could be imposed by certain physical constraints, which 
 %prevent generating all switching sequence in practice,  or by existing controllers which generated the switching sequences.  
%\textbf{Motivating examples}
To illustrate this point, we will consider a number of simple scenarios where the results of the paper could potentially be useful. 
%The discussion below is intended as illustration and we do not claim
% that the presented results will be useful only in this setting. 
  %In fact, even in this setting, it remains future work to demonstrate the usefullness
 %of our approach by investigating several case studies.

\emph{(1) Control and verification of \SLSS\  with switching constraints. }
  \SLSS\ with switching constraints occur naturally in a large number of applications. Such systems arise for example when the supervisory
  logic of the switching law is (partially) fixed. 
  Note that verification or control synthesis of \SLSS\ can be computationally demanding, especially if the properties or control objectives of interest are 
  discrete \cite{CalinBelta1}. The results of the paper could be useful for verification or control of 
  such systems, if the properties of interest or the control objectives depend only on the input-output behavior.
  In this case, we could replace the original \LSS\ $\Sigma$ by the reduced order \LSS $\bar{\Sigma}$ whose  input-output behavior 
  for all the admissible switching sequences coincides with that of $\Sigma$. We  can then perform verification or control synthesis for $\bar{\Sigma}$
  instead of $\Sigma$.  If $\bar{\Sigma}$ satisfies the desired input-output properties, then so does $\Sigma$.
  %For example, in the setting of \cite{CalinBelta}, this would mean that the discrete outputs of the system represent a discretization of
  %the continuous one.
  %If the control objectives or the condition to be verified depends only on the input-output behavior, 
  %the model reduction algorithm yields a \LSS\ whose input-output behavior for admissible switching signals is the same as that of the original DTLSS.
  %Indeed, the reduced \LSS satisfies a certain input-output property for some or all admissible sequences, if and only if the original 
  %\LSS\  satisfied those properties. 
  Likewise, if the composition of $\bar{\Sigma}$ with a controller meets the control objectives, then 
  the composition of this controller with $\Sigma$ meets them too. 

\emph{ (2) Piecewise-affine hybrid systems}.
 Consider a piecewise-linear hybrid system $H$ \cite{BempObs1,HybSys}. Such systems can often be modelled as a feedback interconnection of a linear
 switched system $\Sigma$ of the form \eqref{dtlss1} with a discrete event generator $\phi$, which generates the next discrete state based
 on the past discrete states and past outputs. As a consequence, the solutions of $H$ corresponds to the solutions  $\{q_t,x_t,u_t,y_t\}_{t=0}^{\infty}$ of \eqref{dtlss1} with
 $q_t=\phi(\{y_s,q_s\}_{s=0}^{t-1})$. A simple example of such a system is $q_t=\phi(y_{t-1})$, $t > 0$, and $q_0$ is fixed, where $\phi$ is a piecewise affine map.
 %\begin{equation}
 %\label{pla1}
 % H\left\{\begin{split}
 %    x(t+1) = A_{q(t)}x(t)+B_{q(t)}u(t) \mbox{ and } x(0)=x_0\\
 %     y(t) = C_{q(t)}x(t) \\
 %     q(t) = \phi(y_{t-1},q_{t-1}) \mbox{ and } q(0)=q_0 \\
 % \end{split}\right.
 %\end{equation}
 %where $q(t) \in Q=\{1,\ldots,\QNUM\}$ is the switching signal, $A_q,B_q,C_q$ are matrices of suitable size and $\phi:Q \times \mathbb{R}^p \rightarrow Q$  
 %is a disrcrete-state transition function which defines the conditions for a change in the discrete mode. For the sake of simplicity, we assume that
 %$\phi$ is defined by polyhedral sets, i.e. $\{ y \mid \phi(y,q_1)=q_2\}$ is a polyhedral set for all $q_1,q_2 \in Q$.
 %Such systems arise either by piecewise-linear modelling of a complex plant, or by
 %explicit modelling of a switching controller. It is easy to see that \eqref{pla1} can be viewed as a feedback interconnection of the LSSs 
 %$\Sigma$ of the form \eqref{dtlss1} with the discrete-time controller $\phi$. 
 Often, it is desired to verify if the system is \emph{safe}, i.e., that the sequences of discrete modes generated by the system $H$ belong to a certain set of safe sequences $L$ for all (some) continuous input signals.
% In order to verify safety, all we need to know is the input-output behavior of $\Sigma$ for switching sequences from the set
% $L=\{ q(0)\cdots q(t) \mid q(0)\cdots q(t-1) \in L_s\}$. 
% Consider now the LSSs $\bar{\Sigma}$ from \eqref{dtlss2}, such that
% the input-output behavior of $\Sigma$ and $\bar{\Sigma}$ is the same for all sequences of discrete modes from $L$. 
Consider now another piecewise-affine hybrid system $\bar{H}$ obtained by interconnecting the discrete event generator $\phi$ with  a reduced order \LSS 
$\bar{\Sigma}$, such that
the input-output behavior of $\bar{\Sigma}$ coincides with that of $\Sigma$ for all the switching sequences from $L$. 
%, i.e. the solutions $\{\bar{q}_t,\bar{x}_t,u_t,\bar{y}_t\}_{t=0}^{\infty}$ of $\bar{H}$ satisfy \eqref{dtlss2} and  $\bar{q}_t=\phi(\{\bar{y}_s,\bar{q}_s\}_{s=0}^{t-1})$.
% the 
%  \begin{equation}
% \label{pla2},
%  \bar{H}\left\{\begin{split}
%      bar{x}(t+1) = \bar{A}_{\bar{q}(t)}\bar{x}(t)+\bar{B}_{\bar{q}(t)}u(t) \mbox{ and } x(0)=\bar{x}_0\\
%      \bar{y}(t) = \bar{C}_{\bar{q}(t)}\bar{x}(t) \\
%      \bar{q}(t)=\phi(\bar{y}_{t-1},\bar{q}_{t-1})  \mbox{ and } \bar{q}(0)=q_0 
%  \end{split}\right.
% \end{equation}.
 If $L$ is prefix closed, then $H$ is safe if and only if $\bar{H}$ is safe, and hence it is sufficient to perform safety analysis on 
 $\bar{H}$. Since the number of continuous-states of $\bar{H}$
 is smaller than that of $H$, it is easier to do verification for $\bar{H}$ than for the original model.
 Note that verification of piecewise-affine hybrid systems has high (in certain cases exponential) computational complexity, \cite{PHaver,YordanovBelta2}. 
%Hnce, using a smaller model could potentially enable verification of
% models which used to be intractable in the past. 
%\textbf{3. Control synthesis for piecewise-affine hybrid systems}
Likewise, assume that it is desired to design a control law for $H$ which ensures  that the switching signal generated by 
the closed-loop system belongs to a certain prefix closed set $L$.  
Such problems arise in various settings for hybrid systems \cite{TabuadaBook}.
%Such problems arise for example in abstraction based approach to control. There, one first determines which sequences of discrete modes should be generated by the hybrid system and designs the corresponding discrete event controller, and then one tries to design continous control laws in such a way that the closed-loop system generates the desired discrete state trajectories. 
Again, this problem is solvable for $H$ if and only if it is solvable for $\bar{H}$, and the controller which solves this problem for $\bar{H}$ also solves it for $H$.

\textbf{Related work}
Results on realization theory of
linear switched systems with constrained switching appeared in \cite{MP:BigArticlePartI}. However, \cite{MP:BigArticlePartI} does not yield a model reduction algorithm,
see Remark \ref{remark:priorwork} for a detailed discussion. The algorithm presented in this paper bears a close resemblance to the moment matching method of
\cite{bastug_acc}, but its result and its scope of application are different. 
 The subject of model reduction for hybrid and switched systems was addressed
 in several papers \cite{French1,Zhang20082944,Mazzi1,Chahlaoui,Habets1,China2,China3,Lam1,Kotsalis2,Kotsalis1,6209392,Shaker1}. However, none of them deals with the problem addressed in this paper.

\textbf{Outline}
In Section \ref{sect:prelim}, we fix the notation and terminology of the paper.
In Section \ref{sect:lin_switch_def}, we present the formal definition and main properties of \SLSS. In Section \ref{sect:problem} we give the precise problem statement. In Section \ref{sect:markov}, 
we recall the concept of Markov parameters, and  
we present the fundamental theorem and corollaries which form the basis of the model reduction by moment matching procedure. 
The algorithm itself is stated in Section \ref{sect:alg} in detail.  
Finally, in Section \ref{sect:exam} the algorithm is illustrated on some numerical examples.

\section{Preliminaries: notation and terminology}
\label{sect:prelim}
%%The notation used in this paper is standard in automata theory, see \cite{eilenberg}, \cite{gecseg}.
%%
%%\begin{Notation} \label{Notation1}
%%Denote by $\mathbb{N}$ the set of natural numbers including $0$, and consider a finite non-empty set $Q$ with $D$ elements, which will be called the \emph{alphabet}. Denote by $Q^*$ the set of finite sequences of elements of $Q$. The elements of $Q^*$ are called \emph{strings} or \emph{words} over $Q$. Each non-empty word $w$ is of the form $w=q_1q_2 \cdots q_k$ for some $q_1,q_2,\cdots,q_k \in Q$. The element $q_i$ is called the \emph{$i$th letter of $w$}, for $i=1,2,\cdots,k$ and $k$ is called the \emph{length} of $w$. The \emph{empty sequence (word)} is denoted by $\varepsilon$. The length of word $w$ is denoted by $|w|$; note that $|\varepsilon|=0$. The set of non-empty words is denoted by $Q^+$, i.e. $Q^+=Q^*\backslash\{\varepsilon\}$. The \emph{concatenation} of word $w \in Q^*$ with $v \in Q^*$ is denoted by $wv$.
%%For simplicity, the finite set $Q$ will be identified with its index set, that is $Q=\{1,2,\cdots,D\}$. Moreover, $Q$ will always be endowed with the discrete topology.
%%\end{Notation}

%\begin{Notation} \label{Notation1}
Denote by $\mathbb{N}$ the set of natural numbers including $0$. 
%Denote by $\mathbb{R}_+$ the set
%$[0,+\infty)$ of nonnegative real numbers. The following notation is standard in automata theory, see \cite{ullman}.

Consider a non-empty set $Q$ which will be called the \emph{alphabet}. Denote by $Q^*$ the set of finite sequences of elements of $Q$. The elements of $Q^*$ are called \emph{strings} or \emph{words} over $Q$, and any set $L \subseteq Q^*$ is called a \emph{language} over $Q$. Each non-empty word $w$ is of the form $w=q_1q_2 \cdots q_k$ for some $q_1,q_2,\cdots,q_k \in Q$, $k>0$. In the following, if a word $w$ is stated as $w=q_1q_2 \cdots q_k$, it will be assumed that $q_1,q_2, \dots, q_k \in Q$. The element $q_i$ is called the \emph{$i$th letter of $w$}, for $i=1,2,\dots,k$ and $k$ is called the \emph{length} of $w$. The \emph{empty sequence (word)} is denoted by $\varepsilon$. The length of word $w$ is denoted by $|w|$; note that $|\varepsilon|=0$. The set of non-empty words is denoted by $Q^+$, i.e., $Q^+=Q^*\backslash\{\varepsilon\}$. The set of words of length $k \in \mathbb{N}$ is denoted by $Q^k$. The \emph{concatenation} of word $w \in Q^*$ with $v \in Q^*$ is denoted by $wv$:
if $v=v_{1}v_2 \cdots v_{k}$, and  $w=w_{1}w_2 \cdots w_{m}$, $k > 0, m >0$,
then $vw=v_{1}v_2 \cdots v_{k}w_{1}w_2 \cdots w_{m}$.  If $v=\epsilon$, then $wv=w$; if $w=\epsilon$, then $wv=v$.

If $Q$ has a finite number of elements, say $D$, it will be identified with its index set, that is $Q=\{1,2,\cdots,D\}$.

\section{LINEAR SWITCHED SYSTEMS}
\label{sect:lin_switch_def}

In this section, we present the formal definition of linear switched systems and
recall a number of relevant definitions. We follow the presentation of
\cite{MP:BigArticlePartI,petreczky2013}.
\begin{Definition}[DTLSS] \label{eq:LSSform}
A discrete-time linear switched system (DTLSS) is a tuple
\begin{equation} \label{key}
\Sigma=(p,m,n,Q,\{(A_q,B_q,C_q)|q \in Q\},x_0) 
\end{equation}
where $Q=\{1,\cdots,D\},~D>0,$ called the set of discrete modes, 
$A_q \in \mathbb{R}^{n \times n}$, $B_q \in \mathbb{R}^{n \times m}$, $C_q \in \mathbb{R}^{p \times n}$ are the matrices of the linear system in mode $q \in Q$, and $x_0$ is the initial state. The number $n$ is called the \emph{dimension (order) of $\Sigma$} and will sometimes be denoted by $\text{dim}\Sigma$.
\end{Definition}

% We remark that, with small changes, the definitions presented in this paper also apply to the continuous time case. However, the results relies on the discrete time setting.

\begin{Notation}
In the sequel, we use the following notation and terminology: 
The state space $X = \mathbb R^n$, the output space $Y = \mathbb R^p$, and
the input space $U = \mathbb R^m$. We will write $\overline{ U^+ \times Q^+ }=\{ (u,\sigma) \in U^+ \times Q^+ \mid |u|=|\sigma| \}$,
and $\sigma(t)$ for the $t+1$th element $q_t$ of a
sequence $\sigma=q_1q_2 \cdots q_{|\sigma|} \in Q^{+}$ (the same comment applies to the elements of $U^+$, $X^+$ and $Y^+$).
 \end{Notation} 

Throughout the paper, $\Sigma$ denotes a \LSS of the form \eqref{eq:LSSform}.

\begin{Definition}[Solution]
A \emph{solution} of the \BLSS 
$\Sigma$ at the initial state $x_0 \in X$ and relative to the pair
$(u,\sigma) \in \overline{ U^+ \times Q^+ }$ is a pair $(x, y) \in
X^+ \times Y^+$, $|x|=|\sigma|+1, |y|=|\sigma|$ satisfying
\begin{equation}\label{eq:cs}
\begin{split}
x(t+1) 
&= A_{\sigma(t)} x(t) + B_{\sigma(t)} u(t),~ x(0) = x_0 \\
y(t) &= C_{\sigma(t)} x(t), %~\text{a.e.} 
\end{split}
\end{equation}
for $t=0,1,\ldots,|\sigma|-1$.
\end{Definition}
We shall call $u$ the control
input, $\sigma$ the switching sequence, $x$ the state trajectory, and $y$ the
output trajectory. Note that the pair $(u,\sigma) \in \overline{ U^+ \times Q^+ }$ can be considered as an
input to the DTLSS.

\begin{Definition}[Input-state and input-output maps]
The \emph{input}-\emph{state} map $X_{\Sigma,x_0}$ and
\emph{input-output} map $Y_{\Sigma,x_0}$ for the \BLSS $\Sigma$,
induced by the initial state $x_0 \in X$, are the maps
\begin{align*}
\overline{ U^+ \times Q^+ } \to X^+;
(u,\sigma)\mapsto X_{\Sigma,x_0}(u,\sigma)=x,\\
\overline{ U^+ \times Q^+ } \to Y^+;~ (u,\sigma)\mapsto Y_{\Sigma,x_0}(u,\sigma)=y,
\end{align*}
where $(x,y)$ is the solution of $\Sigma$ at $x_0$ relative to
$(u,\sigma)$.
\end{Definition}

Next, we present the basic system theoretic concepts for \SLSS. 
The input-output behavior of a \LSS realization can be formalized as a map
\begin{equation} \label{eq:inputoutputmap}
f: \overline{ U^+ \times Q^+ } \rightarrow Y^+.
\end{equation}
The value $f(u,\sigma)$ represents the output of the underlying (black-box) system. This system may or may not admit a description by a DTLSS. Next, we define when a \LSS describes (realizes) a map of the form \eqref{eq:inputoutputmap}.

The \LSS $\Sigma$ of the form \eqref{eq:LSSform} is a \emph{realization} of an input-output map $f$ of the form \eqref{eq:inputoutputmap}, if $f$ is the input-output map of $\Sigma$ which corresponds to some initial state $x_0$, i.e., $f=Y_{\Sigma,x_0}$. The map $Y_{\Sigma,x_0}$ will be referred to as the \emph{input-output map of} $\Sigma$, and it will be denoted by $Y_\Sigma$. The following discussion is valid only for realizable input-output maps.

We say that the \SLSS $\Sigma_1$ and $\Sigma_2$ are \emph{equivalent} if $Y_{\Sigma_1}=Y_{\Sigma_2}$.
The \LSS $\Sigma_{\mathrm m}$ is said to be a \emph{minimal} realization of $f$, if 
$\Sigma_{\mathrm m}$ is a realization of $f$, and for any
\LSS $\Sigma$ such that $\Sigma$ is a realization of $f$,
$\dim \Sigma_{\mathrm m} \le \dim \Sigma$. In \cite{MP:BigArticlePartI}, it is stated that a DTLSS realization $\Sigma$ is minimal if and only if it is span-reachable and observable. See \cite{MP:BigArticlePartI} for detailed definitions of span-reachability and observability for LSSs.

\section{Model reduction by restricting the set of admissible sequences of discrete modes} \label{sect:problem}

 In this section, we state formally the problem of restricting the discrete 
 dynamics of the \BLSS.

\begin{Definition} \label{def:NDFA}
  A non-deterministic finite state automaton (NDFA) is a tuple
  $\mathcal{A}=(S,Q,\{\rightarrow_q\}_{q \in Q} ,F, s_0)$ such that 
  \begin{enumerate}
  \item $S$ is the finite state set,
  \item $F \subseteq S$ is the set of accepting (final) states,
  \item $\rightarrow_q \subseteq S \times S$ is the state transition relation
        labelled by $q \in Q$,
  \item $s_0 \in S$ is the initial state.
  \end{enumerate}
  For every $v \in Q^{*}$, define $\rightarrow_v$ inductively as follows:
  $\rightarrow_{\epsilon}=\{ (s,s) \mid s \in S \}$ and 
  $\rightarrow_{vq} = \{ (s_1,s_2) \in S \times S \mid \exists s_3 \in S: (s_1,s_3) \in \rightarrow_{v} \mbox{ and } (s_3,s_2) \in \rightarrow_q \}$ for all $q \in Q$. 
  We denote the fact $(s_1,s_2) \in \rightarrow_v$ by $s_1 \rightarrow_v s_2$. 
  The fact that there exists $s_2$ such that
  $s_1 \rightarrow_v s_2$ is denoted by $s_1 \rightarrow_v$. 
  Define the language $L(\mathcal{A})$ accepted by $\mathcal{A}$ as 
  \[ L(\mathcal{A})=\{ v \in Q^{*} \mid \exists s \in F: s_0 \rightarrow_v s \}. \] 
  %Define the language of all non-empty words accepted by 
  %$\mathcal{A}$ as $L_{+}(\mathcal{A})=L(\mathcal{A}) \setminus \{\epsilon\}$. 
  %Denote by $L_{\omega}(\mathcal{A})$ the 
  %\emph{language of infinite strings generated by $\mathcal{A}$}, i.e.
 %\[ L_{\omega}(\mathcal{A})=\{ w=q_0q_1\cdots  \in \mathcal{Q} \mid \forall i \in \mathbb{N}: q_0q_1\cdots q_i L(\mathcal{A}) \}.
 % \]
\end{Definition}
%\vspace{6pt}
Recall that a language $L \subseteq Q^*$ is regular, if there exists an
  NDFA $\mathcal{A}$ such that $L=L(\mathcal{A})$. In this case, we say that $\mathcal{A}$ \emph{accepts} or \emph{generates} $L$.
  We say that $\mathcal{A}$ is \emph{co-reachable}, if from any state a final state can be reached, i.e., for any $s \in S$, there exists $v \in Q^*$ and 
 $s_f \in F$ such that $s \rightarrow_v s_f$. It is well-known that
 if $\mathcal{A}$ accepts $L$, then we can always compute an NDFA $\mathcal{A}_{co-r}$ from $\mathcal{A}$ such that $\mathcal{A}_{co-r}$ accepts $L$ and it is co-reachable. Hence, without loss of generality, in this paper we will consider only co-reachable NDFAs.
 \begin{Definition}[$L$-realization and $L$-equivalence]
 \label{def:lreal}
  Consider an input-output map $f$ and a \BLSS $\Sigma$. Let $L \subseteq Q^{+}$.  We will say that $\Sigma$ is an $L$-realization of $f$, if
 for every $u \in U^+$, and every $\sigma \in L$ such that $|u|=|\sigma|$,
  \begin{equation}
  \label{problem1:eq2}
     Y_{\Sigma}(u,\sigma)(|\sigma|-1)=f(u,\sigma)(|\sigma|-1),
 \end{equation}
i.e., the final value of $Y_{\Sigma}$ and $f$ agrees for all  $(u,\sigma) \in U^+ \times L$, $|\sigma|=|u|$. Note that a $Q^+$-realization is precisely a realization. We will say that two \BLSS $\Sigma_1$ and $\Sigma_2$ are $L$-equivalent, if
 $\Sigma_2$ is an $L$-realization of $Y_{\Sigma_1}$ (or equivalently if
  $\Sigma_1$ is an $L$-realization of $Y_{\Sigma_2}$).
 \end{Definition}
 \begin{Problem}[Model reduction preserving $L$-equivalence]
 \label{problem1}
  Consider a minimal \BLSS\ $\Sigma$ and let $L \subseteq Q^{+}$ be a regular language. Find a
  \BLSS $\bar{\Sigma}$ such that $\dim \bar{\Sigma} < \dim \Sigma$ and, 
  $\Sigma$ and $\bar{\Sigma}$ are $L$-equivalent.
 \end{Problem}
\begin{Remark}
\label{remark:priorwork} 
 The problem of finding an $L$-realization of
 $f$ was already addressed in \cite{MP:BigArticlePartI,MP:Phd} for the continuous time case. There, it was shown that
 if $\Sigma$ is a realization of $f$ and $M$ is a number which depends on the cardinality of the
 state-space of a deterministic finite state automaton accepting $L$,
 then it is possible to find a $\bar{\Sigma}$ such that
 $\bar{\Sigma}$ is an $L$-realization of $f$ and 
 \begin{equation} 
 \label{problem_form:est}
   \dim \bar{\Sigma} \le M \dim \Sigma. 
 \end{equation}
%Moreover, in \cite{MP:Phd} an algorithm for computing such $\bar{\Sigma}$ was also presented. 
This result may also be extended for the discrete time case in a similar way. However, as \eqref{problem_form:est} shows, the obtained
$L$-realization can even be of higher dimension than the original system. 
\end{Remark}

\section{Model reduction algorithm: preliminaries} \label{sect:markov}
 In order to present the model reduction algorithm and its proof of correctness,
 we need to recall the following definitions from \cite{petrbako}.
 \begin{Definition}[Convolution representation]
 \label{sect:io:def1}
  The input-output map $f$ has a \emph{generalized convolution
  representation (abbreviated as \GCR)}, if there exist maps
  $S^f_0:Q^{+} \rightarrow \mathbb{R}^p$, 
  $S^f:Q^{+} \rightarrow \mathbb{R}^{p \times m}$, such that 
  $S^f(q)=0$ if $q \in Q$ and
  \begin{equation*}
  %\label{sect:io:def1:eq1}
    \begin{split}
     & f(u,\sigma)(t)=S_0^f(q_0q_1\cdots q_t) 
      + \sum_{k=0}^{t-1} S^f(q_kq_{k+1}\cdots q_t)u_k
    \end{split}
  \end{equation*}
  for all $(u,\sigma) \in \overline{ U^+ \times Q^+ }$, $t \leq |\sigma|$ with
  $\sigma=q_0q_1\cdots q_{|\sigma|} $.
 \end{Definition}
%\begin{Definition}[Markov parameters] \label{MarkovParameters}
%The Markov parameters of the input-output map $f$ are the values of the map
%\[ M^f:Q^{*} \rightarrow \mathbb{R}^{\QNUM p \times (m\QNUM + 1)}, \]
%defined by 
%\[ M^f(v)=\begin{bmatrix} 
%         S^f_0(v1) & S^f(1v1) & \cdots & S^f(\QNUM v1) \\
%         S^f_0(v2) & S^f(1v2) & \cdots & S^f(\QNUM v2) \\
%         \vdots  & \vdots   & \cdots & \vdots \\
%         S^f_0(v\QNUM) & S^f(1v\QNUM) & \cdots & S^f(\QNUM v \QNUM) \\
%       \end{bmatrix}.
%\]
%\vspace{5pt}
%\end{Definition}
By a slight abuse of the terminology adopted in \cite{petrbako}, we will call the maps $\{S^f,S_0^f\}$ the \emph{Markov parameters} of $f$.
Notice that if $f$ has a \GCR, then the Markov-parameters of $f$ determine $f$ uniquely.
 In other words, the Markov-parameters of $f$ and $g$ are equal if and only if $f$ and $g$ are the same input-output map, i.e.
  $S_0^{f}=S_0^{g}$ and $S^{f}=S^{g}$ if and only if $f=g$.

In the sequel, we will use the fact that Markov parameters can be expressed via the matrices of a state-space representation. In order to
present this relationship, we introduce the following notation.
\begin{Notation} \label{Notation2}
Let $w=q_1q_2 \cdots q_k \in Q^*$, $k>0$ and $A_{q_i} \in \mathbb{R}^{n \times n}$, $i=1,\cdots,k$. Then the matrix $A_w$ is defined as
\begin{equation} 
A_w=A_{q_k}A_{q_{k-1}}\cdots A_{q_1}.
\end{equation}
If $w= \varepsilon$, then $A_\varepsilon$ is the identity matrix.
% (which will be denoted by $I_n$).
\end{Notation}
 \begin{Lemma}[\cite{petrbako}]
 \label{sect:io:lemma1}
  The map $f$ is  realized by the
  \LSS\  $\Sigma$ if and only if $f$ has a \GCR\ and 
  for all $v \in Q^{*}$, $q,q_0 \in Q$,
  \begin{equation}
  \label{sect:io:lemma1:eq1}
	%\begin{split}
         S^f(q_0vq)=C_qA_vB_{q_0} \quad \mbox{ and } \quad S^{f}_{0}(vq)=C_qA_vx_0.
      %\end{split}
  \end{equation}
  %where $\widetilde{C}=\begin{bmatrix} C_1^{\mathrm{T}} & C_2^{\mathrm{T}} & \cdots & C_{\QNUM}^{\mathrm{T}} \end{bmatrix}^{\mathrm{T}}$.
 \end{Lemma}
 \vspace{6pt}
We will extend Lemma \ref{sect:io:lemma1} to characterize the fact that $\Sigma$ is an $L$ realization of $f$ in terms of Markov parameters. To this end,
we need the following notation.
\begin{Notation}[Prefix and suffix of $L$]
\label{prefix:not}
Let the prefix $(L)_*$ and suffix ${}_*(L)$ of a language $L$ be defined as follows: $(L)_*=\{ s \in Q^{*} \mid \exists w \in Q^*: sw \in L\}$,
and ${}_*(L)=\{ s \in Q^{*} \mid \exists w \in Q^{*}: ws \in L\}$. In addition, let the $1$-prefix $(L)_1$ and $1$-suffix ${}_1(L)$ of a language $L$ be defined as follows: 
$(L)_1=\{ s \in Q^* \mid \exists q \in Q: sq\in L \}$ and ${}_1(L)=\{ s \in Q^* \mid \exists q \in Q: qs \in L \}$. A language $L$ is said to be prefix (resp. suffix) closed if $(L)_*=L$ (resp. ${}_*(L)=L$).
%Note that $(L)_*=\bigcup_{k \in \mathbb{N}}(L)_k$ and ${}_*(L)=\bigcup_{k \in \mathbb{N}}{}_k(L)$. 
\end{Notation}
%Intuitively, this notation for prefix and suffix of a language points to the fact that $k$-prefix (resp. $k$-suffix) of a language $L$ consists of all the words in $L$ with their last (resp. first) $k$ letters removed from the right (resp. left). Subscript $*$ in this sense, denotes removing arbitrary number of letters including none, either from the ending or beginning of the words (depending on the position of $*$).
%and
%and $\mathrm{infix}(L)=\{ s \in Q^{*} \mid \exists w,h \in Q^{*}: wsh \in L\}$.

\begin{Example}
Let the language $L$ be defined as $L= (123)^*12 =\{12, 12312, 12312312, \dots  \}$. Then with the notation stated above, the following languages can be defined as follows:
\begin{equation*}
\begin{aligned}
& {}_*(L)=\{ \varepsilon, 2, 12, 312, 2312, 12312, \dots \} \\
& ({}_*(L))_1=\{ \varepsilon, 1, 31, 231, 1231, \dots \} \\
& (({}_*(L))_1)_*=\{ \varepsilon, 1, 2, 3, 12, 23, 31, 123, 231, 312, \dots \} \\
& (L)_1=\{ 1, 1231, 1231231, \dots \} \\
& ((L)_1)_*=\{ \varepsilon, 1, 12, 123, 1231, 12312, \dots \}.
\end{aligned}
\end{equation*}
\end{Example}
\vspace{6pt}
Note that if $L$ is regular, then so are $(L)_*$, ${}_*(L)$, $(L)_1$, ${}_1(L)$. Moreover 
NDFAs accepting these languages can easily be computed from an NDFA which accepts $L$.

The proof of Lemma \ref{sect:io:lemma1} can be extended to prove the
following result, which will be central for our further analysis.
\begin{Lemma}
 \label{sect:io:lemma2}
Assume $f$ has a GCR.
The \BLSS $\Sigma$ is an $L$-realization of $f$, if and only if for all
$v \in Q^{*}$, $q_0,q \in Q$ 
\begin{equation}
  \label{sect:io:lemma2:eq1}
 \begin{split}
  & vq \in L \implies S_0^f(vq)=C_qA_vx_0  \\
  & q_0vq \in {}_*(L) \implies S^f(q_0vq)=C_qA_vB_{q_0}. \\
 \end{split}
\end{equation}
\end{Lemma}
\vspace{6pt}
Lemma \ref{sect:io:lemma1} implies the following important corollary.
\begin{Corollary}
 \label{sect:io:lemma2:col}
 $\Sigma$ is $L$-equivalent to $\bar{\Sigma}=(p,m,r,Q,\{(\bar{A}_q,\bar{B}_q,\bar{C}_q)|q \in Q\},\bar{x}_0)$   if and only if
\[
  \begin{split}
& (i)\mbox{ } \forall q \in Q, v \in Q^{*}:   vq \in L \implies  C_qA_vx_0 = \bar{C}_q\bar{A}_v\bar{x}_0. \\
& (ii)\mbox{ } \forall q,q_0 \in Q, v \in Q^{*}: q_0vq \in {}_*(L) \implies  C_qA_vB_{q_0} = \bar{C}_q\bar{A}_v\bar{B}_{q_0}. \\
  \end{split}
\]
\end{Corollary}
\vspace{6pt}
That is, in order to find a \BLSS $\bar{\Sigma}$ which is an $L$-equivalent realization of $\Sigma$, it is sufficient to find a \BLSS $\bar{\Sigma}$ which
satisfies parts (i) and (ii) of Corollary \ref{sect:io:lemma2:col}. Intuitively, the conditions (i) and (ii) mean that certain Markov parameters of the input-output maps of $\bar{\Sigma}$ match the corresponding Markov parameters of the input-output map of $\Sigma$. Note that $L$ need not be finite, and hence we might have to match an infinite number of Markov parameters. 
Relying on the intuition of \cite{bastug_acc}, the matching of the Markov parameters can be achieved either by restricting $\Sigma$
to the set of all states which are reachable along switching sequences from  $L$, or by eliminating those states which are not observable for
switching sequences from $L$. Remarkably, these two approaches are each other's dual.

Below we will formalize this intuition. This this end we use the following notation
\begin{Notation}
In the sequel, the image and kernel of a real matrix $M$ are denoted by $\IM (M)$ and $\ker (M)$ respectively. In addition, the $n \times n$ identity matrix is denoted by $I_n$.
\end{Notation}
%We will start with presenting the following definitions.
\begin{Definition}[$L$-reachability space] \label{ReachabilityMat}
For a \LSS $\Sigma$ and $L \subseteq Q^{*}$, define the $L$-reachability space $\mathscr{R}_L(\Sigma)$ as
follows:
\begin{equation}
\begin{split}
& \mathscr{R}_L(\Sigma)= \SPAN \Big \{ \{ A_{v}x_0 \mid v \in Q^*, v \in ((L)_1)_* \} \cup \\
& \{ \IM (A_vB_{q_0}) \mid v \in Q^*, q_0 \in Q, q_0v \in ({}_*((L)_1))_* \} \Big \}.
\end{split}
\end{equation}
Whenever $\Sigma$ is clear from the context, we will denote $\mathscr{R}_L(\Sigma)$ by $\mathscr{R}_L$.  
\end{Definition}
Recall that according to Notation \ref{prefix:not} 
\begin{equation}
\label{ReachabilityMatLanguages}
\begin{split}
 & ((L)_1)_*= \{ s \in Q^* \mid \exists v \in Q^*, \hat{q} \in Q: sv \hat{q} \in L \}, \\
 & ({}_*((L)_1))_* = \{ s \in Q^* \mid \exists v_1,v_2 \in Q^*, \hat{q} \in Q: v_1s v_2 \hat{q}   \in L \}.
\end{split}
\end{equation}
Intuitively, the $L$-reachability space $\mathscr{R}_L(\Sigma)$ of a DTLSS $\Sigma$ is the space consisting of all the states $x \in X$ which are reachable from $x_0$ with some continous input and some switching sequence from $L$.
%such that there exists an input sequence $u \in \mathcal{U}$ and a switching sequence $\sigma \in \mathcal{Q}$ generated by the NDFA $\mathcal{A}$ with $L(\mathcal{A})=L$, starting from its initial state $s_0$ and ending in any of its states $s \in S$, which yields $X_{\Sigma,0}(u,\sigma)=x$. 
It follows from \cite{MP:BigArticlePartI,Sun:Book} that $\Sigma$ is span-reachable if and only if $\dim \mathscr{R}_{Q^*}=n$.
\begin{Definition}[$L$-unobservability subspace] \label{ObservabilityMat}
For a \LSS $\Sigma$, and $L \subseteq Q^*$, define the $L$-unobservability subspace as
\begin{equation} \label{eq:unobs}
\mathscr{O}_L(\Sigma)=\bigcap_{v \in Q^{*}, q \in Q, vq \in {}_*(L)} \ker (C_qA_{v}). 
\end{equation}
If $\Sigma$ is clear from the context, we will denote $\mathscr{O}_L(\Sigma)$ by $\mathscr{O}_L$.
\end{Definition}
Recall that according to Notation \ref{prefix:not}, 
 \begin{equation}
 \label{ObservabilityMatLanguages}
  {}_*(L)=\{ s \in Q^* \mid \exists v \in Q^{*}: vs \in L\}.
 \end{equation}
Intuitively, the $L$-unobservability space $\mathscr{O}_L(\Sigma)$  is the set of all those states which remain unobservable under switching sequences from $L$.
%of a DTLSS $\Sigma$ is the space consisting of all the states $x \in \mathcal{X}$ such that for any switching sequence $\sigma \in \mathcal{Q}$ such that $v$ is a suffix of a sequence from 
%$L(\mathcal{A})=L$, for any input sequence $u \in \mathcal{U}$, $Y_{\Sigma,x}(u,\sigma)=Y_{\Sigma,0}(u,\sigma)$. 
%It is not difficult to see that 
%\[ \mathscr{O}_L(\Sigma)=\bigcap_{i=1}^{\infty} \mathscr{O}_{L,i}, \]
%where
%$\mathscr{O}_{L,0}=\bigcap_{q \in Q, q \in \mathrm{infix}(L)} \ker C_q$ and
%for any $N > 0$, 
%\[ \mathscr{O}_{L,N}=\mathscr{O}_{L,0} \cap \bigcap_{q \in Q} \mathscr{O}_{L_q,N-1}A_q, \]
 %where
%\[ L_q = \{ w \in \mathcal{Q} \mid \exists h \in Q^{*}, hqw \in L \}. \] 

From \cite{Sun:Book}, it follows that  $\Sigma$ is observable if and only if $\mathscr{O}_{Q^*}=\{0\}$. Note that $L$-unobservability space is not defined in a totally ``symmetric'' way to the $L$-reachability space, i.e., subscript of the intersection sign in Equation~\eqref{eq:unobs} is not $vq \in {}_*(_1({}_*(L)))$. See Remarks \ref{rem:lemma3} and \ref{rem:lemma4} for further discussion.

We are now ready to present two results which are central to the model reduction algorithm to be presented in the 
next section.
\begin{Lemma} \label{theo:mert1}
Let $\Sigma=(p,m,n,Q,\{(A_q,B_q,C_q)|q \in Q\},x_0)$ be a \LSS\ and $L \subseteq Q^*$. Let $\dim \mathscr{R}_{L}(\Sigma) =r$ and $P \in \mathbb{R}^{n \times r}$ be a full column rank matrix such that
\[
\mathscr{R}_{L}(\Sigma) = \IM (P).
\]
Let $\bar{\Sigma}=(p,m,r,Q,\{(\bar{A}_q,\bar{B}_q,\bar{C}_q)|q \in Q\},\bar{x}_0)$ be the \LSS defined by
\[
\bar{A}_q=P^{-1}A_qP \mbox{, } \bar{B}_q=P^{-1}B_q \mbox{, } \bar{C}_q=C_qP \mbox{, } \bar{x}_0=P^{-1}x_0,
\]
where $P^{-1}$ is a left inverse of $P$. Then 
$\bar{\Sigma}$ and $\Sigma$ are $L$-equivalent.
%for all $q,q_0 \in Q$, $v \in Q^*$,
%\[
%  \begin{split}
%& (i)\mbox{ } vq \in L \implies  C_qA_vx_0 = \bar{C}_q\bar{A}_v\bar{x}_0 \\
%& (ii)\mbox{ } q_0vq \in {}_*(L) \implies  C_qA_vB_{q_0} = \bar{C}_q\bar{A}_v\bar{B}_{q_0} \\
%  \end{split}
%\]
%i.e., 
\end{Lemma}
That is, Lemma \ref{theo:mert1} says that if we find a matrix representation of the $L$-reachability space, then we
can compute a reduced order \BLSS which is an $L$-realization of $\Sigma$. 

Before presenting the proof of Lemma~\ref{theo:mert1}, we will prove the following claim.

\begin{Claim} \label{claim1}
With the conditions of Lemma~\ref{theo:mert1} the following holds: 

\emph{(i)} For all $v \in Q^*$ such that $v \in ((L)_1)_*$,
\[
\begin{split}
& v=\varepsilon \implies PP^{-1}x_{0}=x_{0}, \\
& v=q_1 \cdots q_k, \mbox{ } k \geq 1 \implies \\
& PP^{-1}A_{q_k} \cdots PP^{-1}A_{q_1}PP^{-1}x_{0}= A_{q_k} \cdots A_{q_1}x_{0}. 
\end{split}
\]

\emph{(ii)} For all $v \in Q^*$, $q_0 \in Q$ such that $q_0v \in (({}_*(L))_1)_*$,
\[
\begin{split}
& v=\varepsilon \implies PP^{-1}B_{q_0}=B_{q_0}, \\
& v=q_1 \cdots q_k, \mbox{ } k \geq 1 \implies \\ 
& PP^{-1}A_{q_k} \cdots PP^{-1}A_{q_1}PP^{-1}B_{q_0}= A_{q_k} \cdots A_{q_1}B_{q_0}.
\end{split}
\]
\end{Claim}
%\vspace{6pt}
\begin{proof}\emph{(Claim~\ref{claim1} (ii))}
The proof is by induction on the length of $v$. For $|v|=0$, let $q_0 \in Q$ and $q_0v \in(({}_*(L))_1)_*$. The assumption $\mathscr{R}_L(\Sigma)= \IM (P)$ in Lemma~\ref{theo:mert1} implies $\IM (B_{q_0}) \subseteq \IM (P)$. Hence, there exists an $\Lambda \in \mathbb{R}^{r \times m}$ such that $P \Lambda=B_{q_0}$, and therefore $PP^{-1}B_{q_0}=PP^{-1}P \Lambda= P \Lambda=B_{q_0}$.

%For $|v|=1$, let $v=q_1$; $q_0,q_1 \in Q$ and $q_0v \in \mathrm{rightshift}(\mathrm{suffix}(L))$. The condition $\mathscr{R}_L(\Sigma)= \IM P$ again implies that $\IM A_{q_1}B_{q_0} \subseteq \mathscr{R}_L=\IM P$. Thus
%\[
%PP^{-1}A_{q_1}PP^{-1}B_{q_0}=PP^{-1}A_{q_1}B_{q_0}=A_{q_1}B_{q_0}.
%\]

For $|v|=k \geq 1$, $k \in \mathbb{N}$, let $v=q_1 \cdots q_k$, $q_0v \in (({}_*(L))_1)_*$. Observe that if $q_0v \in (({}_*(L))_1)_*$ then also $q_0\hat{v} \in (({}_*(L))_1)_*$ where $\hat{v}=q_1 \cdots q_{k-1}$, since the set $(({}_*(L))_1)_*$ is prefix closed. Assume the claim holds for $|v|=k-1$, i.e., for $v=q_1 \cdots q_{k-1}$. Then
\begin{multline*}
PP^{-1}A_{q_k}PP^{-1}A_{q_{k-1}} \cdots PP^{-1}A_{q_1}PP^{-1}B_{q_0}= \\ 
PP^{-1}A_{q_k}A_{q_{k-1}} \cdots A_{q_1}B_{q_0}.
\end{multline*}
Since again $\IM (A_{q_k} \cdots A_{q_1}B_{q_0}) \subseteq \IM (P)$, it follows that
\[
PP^{-1}A_{q_k}A_{q_{k-1}} \cdots A_{q_1}B_{q_0}=A_{q_k}A_{q_{k-1}} \cdots A_{q_1}B_{q_0},
\]
proving this part. The proof of part \emph{(i)} is similar.
\end{proof}

\begin{proof}[Proof of Lemma \ref{theo:mert1}] 
 %By applying Lemma \ref{sect:io:lemma2} to $f=Y_{\Sigma}$ and using Lemma \ref{sect:io:lemma1}, it follows that it is enough to show that for all
We will show that part (i) and (ii) of Corollary \ref{sect:io:lemma2:col} hold.

\emph{(ii)}. Using Claim~\ref{claim1} \emph{(ii)} and observing $({}_*(L))_1 \subseteq (({}_*(L))_1)_*$, it follows that, for all $v \in Q^*$, $q_0,q \in Q$ such that $q_0vq \in {}_*(L)$
\[
\begin{split}
& v=\varepsilon \implies \bar{C}_q\bar{B}_{q_0}= C_q PP^{-1}B_{q_0}=C_qB_{q_0}, \\
& v=q_1 \cdots q_k, \mbox{ } k \geq 1 \implies \\ 
& \bar{C}_q\bar{A}_v\bar{B}_{q_0}= C_qPP^{-1}A_{q_k}PP^{-1}A_{q_{k-1}} \cdots PP^{-1}A_{q_1}PP^{-1}B_{q_0}= \\
& C_qA_{q_k} A_{q_{k-1}} \cdots A_{q_1}B_{q_0}.
\end{split}
\]

\emph{(i)} Similar to part \emph{(ii)}.
%Using Claim~\ref{claim1} \emph{(i)} and observing $(L)_1 \subseteq ((L)_1)_*$, it follows that, for all $v \in Q^*$, $q \in Q$ such that $vq \in L$ 
%\[
%\begin{split}
%& v=\varepsilon \implies \bar{C}_q\bar{x}_0= C_q PP^{-1}x_{0}=C_q x_{0}, \\
%& v=q_1 \cdots q_k, \mbox{ } k \geq 1 \implies \\ 
%& \bar{C}_q\bar{A}_v\bar{x}_0= C_qPP^{-1}A_{q_k}PP^{-1}A_{q_{k-1}} \cdots PP^{-1}A_{q_1}PP^{-1}x_{0}= \\
%& C_qA_{q_k} A_{q_{k-1}} \cdots A_{q_1}x_{0}.
%\end{split}
%\]
\end{proof}
By similar arguments we also obtain:
\begin{Lemma} \label{theo:mert2}
Let $\Sigma=(p,m,n,Q,\{(A_q,B_q,C_q)|q \in Q\},x_0)$ be a \LSS and let $L \subseteq Q^*$. Let $\mathrm{codim} \mathscr{O}_L({\Sigma)}=r$ and $W \in \mathbb{R}^{r \times n}$ be a full row rank matrix such that
\[
   \mathscr{O}_{L}(\Sigma) = \ker (W).
\]
Let $\bar{\Sigma}=(p,m,r,Q,\{(\bar{A}_q,\bar{B}_q,\bar{C}_q)|q \in Q\},\bar{x}_0)$ be the \LSS defined by
\[
\bar{A}_q=WA_qW^{-1} \mbox{, } \bar{B}_q=WB_q \mbox{, } \bar{C}_q=C_qW^{-1} \mbox{, } \bar{x}_0=Wx_0.
\]
where $W^{-1}$ is a right inverse of $W$. Then $\bar{\Sigma}$ is an $L$-equivalent to $\Sigma$. 
%for all $q,q_0 \in Q$, $v \in Q^*$,
%\[
%\begin{split}
%& \mbox{(i) } vq \in L \implies  C_qA_vx_0 = \bar{C}_q\bar{A}_v\bar{x}_0 \\
%& \mbox{(ii) } q_0vq \in {}_*(L) \implies  C_qA_vB_{q_0} = \bar{C}_q\bar{A}_v\bar{B}_{q_0} \\
%\end{split}
%\]
\end{Lemma}

\begin{Remark} \label{rem:lemma3}
\emph{(Lemma~\ref{theo:mert1})} Observe that from Lemma~\ref{sect:io:lemma2}, the only Markov parameters involved in the output at a final state of the NDFA are of the form $C_qA_vx_0$ where $v \in (L)_1$ and $C_qA_vB_{q_0}$ where $q_0v \in ({}_*(L))_1$. However, for the induction in the proof of Claim~\ref{claim1} to work out, it is crucial that the words $v$, $q_0v$ which are indexing the elements of the space $\SPAN \{A_vx_0, \IM (A_vB_{q_0}) \}$, must belong to prefix closed sets. Since the smallest prefix-closure of a language $K$ must be $(K)_*$, the prefix-closure of the sets $(L)_1$ and $({}_*(L))_1$ are used in the definition of $L$-reachability space; namely the sets $((L)_1)_*$ and $(({}_*(L))_1)_*$ respectively. This fact leads to matching all the outputs (or the Markov parameters involved in the output) of the original and reduced order systems generated in the course of a switching sequence from $L$, as opposed to matching only the final outputs. The latter would be sufficient for $L$-equivalence.
%instances before the NDFA reaches one of its final states. Hence, the resulting DTLSS is more than just an $L$-realization. 
%\begin{equation*}
%\begin{split}
%& \mathscr{R}_L(\Sigma)= \mathrm{Span}\{ \{ A_{v}x_0 \mid v \in Q^*, v \in ((L)_*)_1 \}
%, \\
%& \{ \IM (A_vB_{q_0}) \mid v \in Q^*, q_0 \in Q, q_0v \in (({}_*(L))_*)_1 \} \}. 
%\end{split}
%\end{equation*}
%Since $((K)_1)_*=((K)_*)_1$ for a regular language $K$, our formulation for $L$-reachability space and Lemma~\ref{theo:mert1} present in the paper, happens to coincide with this formulation (related to matching \emph{all} output instances of the original and reduced order systems before the NDFA reaches one of its final states).
\end{Remark}

\begin{Remark} \label{rem:lemma4}
\emph{(Lemma~\ref{theo:mert2})} Observe that from Lemma~\ref{sect:io:lemma2}, the only Markov parameters involved in the output at a final state of the NDFA are of the form $C_qA_vx_0$ and $C_qA_vB_{q_0}$ where $vq \in {}_*(L)$. In addition, for the induction in the proof of the counterpart of Claim~\ref{claim1} to function (in the case of Lemma~\ref{theo:mert2}), it suffices that the words $vq$ which are indexing the elements of the space $\bigcap \ker (C_qA_v)$ belong to a suffix-closed set. Since ${}_*(L)$ is already suffix-closed, in this case there is no need to expand the set ${}_*(L)$ for the definition of $L$-unobservability space. Hence the reduced order system found by the use of Lemma~\ref{theo:mert2} will be an $L$-realization, but it need not be anything more. 
%In other words, only the Markov parameters involved in the output of the DTLSS at the instances when the NDFA reaches precisely to a final state are guaranteed to be the same for the original and reduced order systems. 
These will be illustrated in the last section with numerical examples. 
\end{Remark}

\section{Model reduction algorithm} \label{sect:alg}

In this section, we present an algorithm for solving Problem \ref{problem1}. The proposed algorithm relies on computing the matrices $P$ and $W$ which satisfy the conditions of Lemma \ref{theo:mert1} and Lemma \ref{theo:mert2} respectively. 
In order to compute these matrices, we will formulate alternative definitions of $L$-reachability/unobservability spaces. To this end, 
for matrices $G,H$ of suitable dimensions and for  a regular language $K  \subseteq Q^{*}$
define the sets $\mathscr{R}_{K}(G)$ and $\mathscr{O}_{K}(G)$ as follows;
\begin{align*}
& \mathscr{R}_{K}(G)=\SPAN \{ \IM (A_vG) \mid v \in K \} \\
& \mathscr{O}_{K}(H)=\bigcap_{v \in K} \ker (HA_v).
\end{align*}
%In addition, for a specific $q \in Q$, define the languages ${}^q(K)$ and $(K)^q$ as
%\begin{align*}
%& {}^q(K)=\{ s \in Q^* \mid qs \in K \} \\
%& (K)^q=\{ s \in Q^* \mid sq \in K \}.
%\end{align*}
Then the $L$-reachability space of $\Sigma$ can be written as
\begin{equation} \label{eq:altreach}
\mathscr{R}_L=\mathscr{R}_{((L)_1)_*}(x_0) + \sum_{q \in Q} \mathscr{R}_{({}^q(K))}(B_q),
\end{equation}
where $((L)_1)_*$ is defined as in \eqref{ReachabilityMatLanguages}, and
\begin{equation}
 \label{eq:altreach:eq1}
{}^q(K)=\{ s \in Q^* \mid \exists v_1,v_2 \in Q^*, \hat{q} \in Q, v_1 q s v_2 \hat{q}   \in L \}.
\end{equation}
%% \item
%%   \[ \mathscr{O}_L = (\sum_{q \in Q} \mathscr{R}(C_q^{\mathrm{T}}, (\stackrel{\leftarrow}{(\mathrm{infix}(L)})_q),\{A_q^{\mathrm{T}}\}_{q \in Q})^{\bot} \] 
%% \end{enumerate}
In \eqref{eq:altreach}, $+$ and $\sum$ denote sums of subspaces, i.e. if $\mathcal{W},\mathcal{V}$ are two linear subspaces of $\mathbb{R}^n$, then
$\mathcal{W}+\mathcal{V}=\{ a+b \mid a \in \mathcal{W}, b \in \mathcal{W}\}$. Similarly, if $\{\mathcal{W}_i\}_{i \in I}$ is a finite family of
linear subspaces of $\mathbb{R}^n$, then $\sum_{i \in I} \mathcal{W}_i=\{ \sum_{i \in I} a_i \mid a_i \in \mathcal{W}_i, i \in I\}$. 

The $L$-unobservability space can be written as
\begin{equation} \label{eq:altobs}
\mathscr{O}_L= \bigcap_{q \in Q} \mathscr{O}_{((K)^q)}(C_q),
\end{equation}
where 
\begin{equation}
\label{eq:tobs:eq1}
(K)^q=\{ s \in Q^* \mid \exists v \in Q^*, vsq \in L \}.
\end{equation}
Note that if $L$ is regular, then ${}^q(K)$ and $(K)^q$, $q \in Q$, and $((L)_1)_*$ are also regular and NDFA's accepting ${}^q(K)$, $(K)^q$, $q \in Q$, and $((L)_1)_*$ can easily be computed from an NDFA accepting $L$. 
%% \item
%%   \[ \mathscr{O}_L = (\sum_{q \in Q} \mathscr{R}(C_q^{\mathrm{T}}, (\stackrel{\leftarrow}{(\mathrm{infix}(L)})_q),\{A_q^{\mathrm{T}}\}_{q \in Q})^{\bot} \] 
%% \end{enumerate}
From \eqref{eq:altreach} and \eqref{eq:altobs} it follows that in order to compute the matrix $P$ in Lemma~\ref{theo:mert1} or $W$ in Lemma~\ref{theo:mert2}, it is enough to compute representations of the subspaces $\mathscr{R}_K(G)$ and $\mathscr{O}_K(H)$ for various choices of $K$, $G$ and $H$.
The corresponding algorithms are presented in Algorithm \ref{alg0} and Algorithm \ref{alg2}. There, we used the following notation.
\begin{Notation}[\textbf{orth}]
For a matrix $M$, $\mathbf{orth}(M)$ denotes the matrix $U$ whose columns form an orthogonal basis of $\IM(M)$.
\end{Notation}
%we present an algorithm for computing a matrix $\hat{P}$ such that $\IM(\hat{P})=\mathscr{R}_K(G)$, if $K=L(\hat{\mathcal{A}})$ for some NDFA $\hat{\mathcal{A}}$.
% full column rank matrix $U$ such that, $\Rank U=\Rank M$, $\IM U=\IM M$ and $U^{\texttt{T}}U=I$.
\begin{algorithm}
\caption{
         Calculate  a matrix representation of $\mathscr{R}_K(G)$, 
         \newline
         \textbf{Inputs}: $(\{A_q\}_{q \in Q},G)$ and $\hat{\mathcal{A}}=(S,\{\rightarrow_q \}_{q \in Q},F,s_0)$ such that $L(\hat{\mathcal{A}})=K$, $F=\{s_{f_1},\cdots s_{f_k}\}$, $k \geq 1$ and 
         $\hat{\mathcal{A}}$  is co-reachable.
         \newline
         \textbf{Outputs:} $\hat{P}  \in \mathbb{R}^{n \times \hat{r}}$ such that $\hat{P}^{\mathrm{T}}\hat{P}=I_{\hat{r}}$, $\Rank(\hat{P})=\hat{r}$,
                           $\IM (\hat{P}) = \mathscr{R}_K(G)$. 
}
\label{alg0}
\begin{algorithmic}[1]
\STATE $\forall s \in S \backslash \{s_0\}: P_s:=0$.
\STATE $P_{s_0}:=\mathbf{orth}(G)$.
\label{alg0.0}
\STATE $\mathrm{flag}=0$.
%\FOR{$ \in S, \in Q: s_0 \rightarrow_q s$}
% \STATE $P_{s}:=\begin{bmatrix} P_{s}, & B_q \end{bmatrix}$
%\ENDFOR
%\STATE $\forall s \in S: P_{s} := \mathbf{orth}(P_{s})$.
\WHILE{$\mathrm{flag}=0$}
\label{alg0.1}
 \STATE $\forall s \in S: P_s^{old} := P_s$
 \FOR{$s \in S$}
  \STATE  $W_s:=P_s$
   \FOR{$q \in Q, s^{'} \in S: s^{'} \rightarrow_q s$}
    \STATE $W_s:=\begin{bmatrix} W_s, & A_qP^{old}_{s^{'}} \end{bmatrix}$
   \ENDFOR
  \STATE $P_s := \mathbf{orth}(W_s)$
 \ENDFOR
 \IF{$\forall s \in S: \Rank (P_s) = \Rank (P^{old}_s$)}
  \STATE{$\mathrm{flag}=1$.}
 \ENDIF 
\ENDWHILE
\RETURN $\hat{P}=\mathbf{orth} \left( \begin{bmatrix} P_{s_{f_1}} & \cdots & P_{s_{f_k}} \end{bmatrix} \right)$.
\end{algorithmic}
\end{algorithm}

\begin{Lemma}[Correctness of Algorithm \ref{alg0} -- Algorithm \ref{alg2}]
 Algorithm \ref{alg0} computes $\mathscr{R}_K(G)$ and Algorithm \ref{alg2} computes $\mathscr{O}_K(H)$.
\end{Lemma}

\begin{proof}
 We prove only the first statement of the lemma, the second one can be shown using duality.
Let
$P_{s,i}=\SPAN\{ \IM (A_vG) \mid v \in Q^*, |v| \le i, s_0 \rightarrow_v s\}$. 
It then follows that after the execution of Step \ref{alg0.0},
$\IM (P_s)=P_{s,0}$ for all $s \in S$. Moreover, by induction it follows that
\[ P_{s,i+1}=P_{s,i} + \sum_{q \in Q, s^{'} \in S, s^{'} \rightarrow_q s} A_qP_{s^{'},i} \]
for all $i=0,1,\ldots$ and $s \in S$. Hence, by induction it follows that
at the $i$th iteration of the loop in Step \ref{alg0.1}, $\IM (P_s)=P_{s,i}$.
Notice that
$P_{s,i} \subseteq P_{s,i+1} \subseteq \mathbb{R}^n$ and hence there exists
$k_s$ such that $P_{s,k_s}=P_{s,k}$, $k \ge k_s$, and thus $P_{s,k}=R_s$,
\[ R_s=\SPAN\{  \IM (A_vG) \mid v \in Q^{*}, s_0 \rightarrow_v s \}.\]
Let $k= \max \{k_s |  s \in S\}$.
It then follows that $P_{s,k+1}=P_{s,k}=\IM (P_s)$ for all $s \in Q$ and hence
after $k$ iterations, the loop \ref{alg0.1} will terminate. 
Moreover, in that case,  $\IM (P_{s_{f_i}})=R_{s_{f_i}}$, $i \in \{1,\cdots,k\}$. But notice that
for any $v \in Q^*$, $q \in Q$, 
$s_0 \rightarrow_v s_{f_i}$ if and only if $v \in K$, and
$s_0 \rightarrow_{qv} s_{f_i}$ if and only if $qv \in K$, $i \in \{1,\cdots,k\}$. Hence,
$\sum\limits_{s \in F}^{}R_s=\mathscr{R}_K$  and thus $\IM \left( \left[ \begin{array}{ccc} P_{s_{f_1}} & \cdots & P_{s_{f_k}} \end{array} \right] \right)=\mathscr{R}_K$.
\end{proof}

\begin{algorithm}
\caption{
         Calculate a matrix representation of $\mathscr{O}_K(H)$, 
         \newline
         \textbf{Inputs}: $(\{A_q\}_{q \in Q},H)$ and $\hat{\mathcal{A}}=(S,\{\rightarrow_q \}_{q \in Q},F,s_0)$ such that $L(\hat{\mathcal{A}})=K$, $F=\{s_{f_1},\cdots s_{f_k}\}$, $k \geq 1$ and 
         $\hat{\mathcal{A}}$  is co-reachable.
         \newline
         \textbf{Outputs:} $\hat{W}  \in \mathbb{R}^{\hat{r} \times n}$ such that $\hat{W}\hat{W}^{\mathrm{T}}=I_{\hat{r}}$, $\Rank(\hat{W})=\hat{r}$,
                           $\ker(\hat{W}) = \mathscr{O}_K(H)$. 
}
\label{alg2}
\begin{algorithmic}[1]
\STATE $\forall s \in S \backslash F: W_s:=0$.
\STATE $\forall s \in F: W_{s}^{\mathrm{T}}:=\mathbf{orth}(H^{\mathrm{T}})$.
\STATE $\mathrm{flag}=0$.
%\FOR{$ \in S, \in Q: s_0 \rightarrow_q s$}
% \STATE $P_{s}:=\begin{bmatrix} P_{s}, & B_q \end{bmatrix}$
%\ENDFOR
%\STATE $\forall s \in S: P_{s} := \mathbf{orth}(P_{s})$.
\label{alg2.0}
\WHILE{$\mathrm{flag}=0$}
\label{alg2.1}
 \STATE $\forall s \in S: W_s^{old} := W_s$
 \FOR{$s \in S$}
  \STATE  $P_s:=W_s$
   \FOR{$q \in Q, s^{'} \in S: s \rightarrow_q s^{'}$}
    \STATE $P_s:=\begin{bmatrix} P_s \\ W^{old}_{s^{'}}A_q \end{bmatrix}$
   \ENDFOR
  \STATE $W_s^{\mathrm{T}} := \mathbf{orth}(P_s^{\mathrm{T}})$
 \ENDFOR
 \IF{$\forall s \in S: \Rank (W_s) = \Rank (W^{old}_s)$}
  \STATE{$\mathrm{flag}=1$.}
 \ENDIF
\ENDWHILE
\RETURN $\hat{W}=W_{s_0}$.
\end{algorithmic}
\end{algorithm}
Notice that the computational complexities of Algorithm \ref{alg0} and Algorithm \ref{alg2}
are polynomial in $n$, even though the spaces of $\mathscr{R}_L$ (resp. $\mathscr{O}_L$) might be generated by images (resp. kernels) of exponentially many matrices.
%\rule{0.46\textwidth}{1pt}
\begin{algorithm}
\caption{Reduction for \SLSS
  \newpage 
  \textbf{Inputs:} $\Sigma=(p,m,n,Q,\{(A_q,B_q,C_q)|q \in Q\},x_0)$ and $\mathcal{A}=(S,\{\rightarrow_q \}_{q \in Q},F,s_0)$ such that $L(\mathcal{A})=L$, $F=\{s_{f_1},\cdots s_{f_k}\}$, $k \geq 1$ and $\mathcal{A}$  is co-reachable.
  \newpage  
  \textbf{Output: } $\bar{\Sigma}=(p,m,r,Q,\{(\bar{A}_q,\bar{B}_q,\bar{C}_q)|q \in Q\},\bar{x}_0)$.
}
\label{alg3}
\begin{algorithmic}[1]
\STATE Compute a co-reachable NDFA $\hat{\mathcal{A}}_r$ from $\mathcal{A}$ such that $L(\hat{\mathcal{A}}_r)=((L)_1)_*$, where $((L)_1)_*$ is as in \eqref{ReachabilityMatLanguages}.
\STATE Use Algorithm \ref{alg0} with inputs $(\{A_q\}_{q \in Q},x_0)$ and NDFA $\hat{\mathcal{A}}_r$.  Store the output $\hat{P}$ as $P_{x_0}:=\hat{P}$.
\FOR{$q \in Q$}
 \STATE  Compute a co-reachable NDFA $\hat{\mathcal{A}}_{r,q}$ from $\mathcal{A}$ such that $L(\hat{\mathcal{A}}_{r,q})={}^q(K)$, where ${}^q(K)$ is as in \eqref{eq:altreach:eq1}.
 \STATE  Use Algorithm \ref{alg0} with inputs $(\{A_q\}_{q \in Q},B_q)$ and NDFA $\hat{\mathcal{A}}_{r,q}$  Store the output $\hat{P}$ as $P_{q}:=\hat{P}$.
\ENDFOR
\STATE $P=\mathbf{orth}(\begin{bmatrix} P_{x_0} & P_1 & \cdots & P_D \end{bmatrix})$
\FOR{$q \in Q$}
 \STATE Compute a co-reachable NDFA $\hat{\mathcal{A}}_{o,q}$ from $\mathcal{A}$, such that $L(\hat{\mathcal{A}}_{o,q})=(K)^q$, 
       where $(K)^q$ is as in \eqref{eq:tobs:eq1}.
 \STATE  Use Algorithm \ref{alg2} with inputs $(\{A_q\}_{q \in Q},C_q)$ and NDFA $\hat{\mathcal{A}}_{o,q}$.  Store the output $\hat{W}$ as $W_{q}:=\hat{W}$.
\ENDFOR
\STATE $W^{\mathrm{T}}=\mathbf{orth}(\begin{bmatrix} W_{1}^{\mathrm{T}} & \cdots & W_D^{\mathrm{T}} \end{bmatrix})$
\IF{$\Rank  (P) < \Rank (W)$}
%$3.ii$: For each $q \in Q$, define the matrices $\bar{A}_q \in \mathbb{R}^{r_1 \times r_1}$, $\bar{B}_q \in \mathbb{R}^{r_1 \times m}$, $\bar{C}_q \in \mathbb{R}^{p \times r_1}$ and the vector $\bar{x}_0 \in \mathbb{R}^{r_1}$ by
\STATE
 Let $r=\Rank (P)$, $P^{-1}$ be a left inverse of $P$ and set
\[
\bar{A}_q=P^{-1}A_qP \mbox{, } \bar{C}_q=C_qP \mbox{, } \bar{B}_q=P^{-1}B_q \mbox{, } \bar{x}_0=P^{-1}x_0.
\]
\ENDIF
\IF{$\Rank (P) \geq \Rank (W)$}
\STATE
 Let $r=\Rank (W)$ and let $W^{-1}$ be a right inverse of $W$. Set
\[
\bar{A}_q=WA_qW^{-1} \mbox{, } \bar{C}_q=C_qW^{-1} \mbox{, } \bar{B}_q=WB_q \mbox{, } \bar{x}_0=Wx_0.
\]
\ENDIF
\RETURN $\bar{\Sigma}=(p,m,r,Q,\{(\bar{A}_q,\bar{B}_q,\bar{C}_q)|q \in Q\},\bar{x}_0)$.
\end{algorithmic}
\end{algorithm}
%\rule{0.49\textwidth}{1pt}

Using Algorithm \ref{alg0} and \ref{alg2}, we can state Algorithm \ref{alg3} for solving Problem \ref{problem1}.
The matrices $P$ and $W$ computed in Algorithm \ref{alg3} satisfy the conditions of Lemma \ref{theo:mert1} and Lemma \ref{theo:mert2} respectively. 
Lemma \ref{theo:mert1} -- \ref{theo:mert2} then imply the following corollary.
\begin{Corollary}[Correctness of Algorithm \ref{alg3}]
 The LSS $\bar{\Sigma}$ returned by Algorithm \ref{alg3} is a solution of Problem \ref{problem1}, i.e. $\bar{\Sigma}$ is $L$-equivalent to $\Sigma$ and $\dim \bar{\Sigma} \le \dim \Sigma$.
 %$L$-realization of $f=Y_{\Sigma}$. Moreover, the dimension of the returned
 %\BLSS does not exceed the dimension of $\Sigma$. 
\end{Corollary}

\section{NUMERICAL EXAMPLES}
\label{sect:exam}

In this section, the model reduction method for DTLSSs with restricted discrete dynamics will be illustrated by $2$ numerical examples. The data used for both examples and \textsc{Matlab} codes for the algorithms stated in the paper are available online from https://kom.aau.dk/\texttildelow mertb/. In the first example, 
it turns out that the rank of the $P$ matrix from Lemma \ref{theo:mert1} is less than the rank of $W$ from Lemma \ref{theo:mert2}. In the second example, the opposite is the case.  For both examples, we used the same NDFA from Figure~\ref{fig:automaton} to define the set of admissible switching sequences: the NDFA is defined as the tuple $\mathcal{A}=(S,\{\rightarrow_q \}_{q \in Q},F,s_0)$ where $S=\{s_0,s_1,s_f\}$, $\rightarrow_1=\{(s_0,s_1)\}$, $\rightarrow_2=\{(s_1,s_f)\}$, $\rightarrow_3=\{(s_f,s_0)\}$ and $F=\{s_f\}$ .

\begin{figure}[H]
\centering
\includegraphics[width=0.2\textwidth]{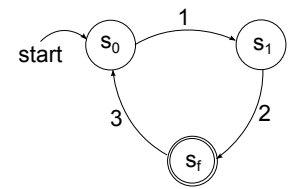}
\caption{The NDFA $\mathcal{A}$ accepting the switching sequence language for both examples.}
\label{fig:automaton}
\end{figure}

Observe that the language $L$ accepted by the NDFA $\mathcal{A}$ is the set $L=\{ 12, 12312, \dots \}$ and it can also be represented by the regular expression $L=(123)^*12$. As stated in Definition~\ref{def:NDFA}, the labels of the edges represent the discrete mode indices of the DTLSS. The parameters of the single input - single output (SISO) DTLSS $\Sigma$ of order $n=7$ with $Q=\{ 1,2,3 \}$ used for the first example are of the following form

\footnotesize
\begin{align*}
& \texttt{A}_{\texttt{1}}= \left[\texttt{zeros(7,1)}, \mbox{ } \texttt{randn(7,2)}, \mbox{ } \texttt{zeros(7,1)}, \mbox{ } \texttt{randn(7,3)}\right] \\
& \texttt{A}_{\texttt{2}}= \left[\texttt{randn(7,1)}, \mbox{ } \texttt{zeros(7,1)}, \mbox{ } \texttt{randn(7,5)}\right] \\
& \texttt{A}_{\texttt{3}}= \left[\texttt{randn(7,2)}, \mbox{ } \texttt{zeros(7,1)}, \mbox{ } \texttt{randn(7,4)}\right] \\
& \texttt{B}_{\texttt{1}}= \left[\texttt{0}, \mbox{ } \texttt{1}, \mbox{ } \texttt{0}, \mbox{ } \texttt{0}, \mbox{ } \texttt{0}, \mbox{ } \texttt{0}, \mbox{ } \texttt{0} \right]^{\texttt{T}} \\
& \texttt{B}_{\texttt{2}}= \left[\texttt{0}, \mbox{ } \texttt{0}, \mbox{ } \texttt{1}, \mbox{ } \texttt{0}, \mbox{ } \texttt{0}, \mbox{ } \texttt{0}, \mbox{ } \texttt{0} \right]^{\texttt{T}} \\
& \texttt{B}_{\texttt{3}}= \left[\texttt{0}, \mbox{ } \texttt{0}, \mbox{ } \texttt{0}, \mbox{ } \texttt{1}, \mbox{ } \texttt{0}, \mbox{ } \texttt{0}, \mbox{ } \texttt{0} \right]^{\texttt{T}} \\
& \texttt{C}_{\texttt{1}}=\texttt{randn(1,7)} \\
& \texttt{C}_{\texttt{2}}=\texttt{randn(1,7)} \\
& \texttt{C}_{\texttt{3}}=\texttt{randn(1,7)} \\
& \texttt{x}_{\texttt{0}}=\left[\texttt{1}, \mbox{ } \texttt{0}, \mbox{ } \texttt{0}, \mbox{ } \texttt{0}, \mbox{ } \texttt{0}, \mbox{ } \texttt{0}, \mbox{ } \texttt{0} \right]^{\texttt{T}}
\end{align*}
\normalsize
where \texttt{randn} and \texttt{zeros} are the \textsc{Matlab} functions which generates arrays containing random real numbers with standard normal distribution and zeros respectively. Applying Algorithm~\ref{alg3} to this DTLSS whose admissible switching sequences are generated by the NDFA shown in Figure~\ref{fig:automaton}, yields a reduced order system $\bar{\Sigma}$ of order $r=4$, whose output values are the same as the original system $\Sigma$ along the allowed switching sequences. This corresponds to an $L$-realization in the sense of Definition \ref{def:lreal} since the language $L$ of the NDFA $\mathcal{A}$ is defined as the set of all words generated by $\mathcal{A}$ starting from its initial state and ending in a final state. In this example, it turns out the algorithm makes use of Lemma~\ref{theo:mert1} and constructs the $P$ matrix. The $P \in \mathbb{R}^{n \times r}$ matrix acquired is $P= \left[ \begin{array}{c} I_4 \\ \hline \mathbf{0} \end{array} \right]$.
%\begin{equation*}
%P= \left[ \begin{array}{c} I_4 \\ \hline \mathbf{0} \end{array} \right].
%\end{equation*}
Note that as stated in Remark~\ref{rem:lemma3}, the resulting DTLSS is more than just an $L$-realization of $\Sigma$, its output coincides with the output of $\Sigma$ for all instances along the allowed switching sequences, rather than just the instances corresponding to the final states of the NDFA (the switching sequence generated by $\mathcal{A}$ used for simulating the examples is given in \eqref{eq:ex_switchsig}). This fact is visible from Figure~\ref{fig:output_ex1} where it can be seen that output of both systems corresponding to all instances along the switching sequence $\sigma$ of length $|\sigma|=11$ defined by
\begin{equation} \label{eq:ex_switchsig}
\sigma=12312312312
\end{equation}
coincide (the input sequence of length $11$ used in the simulation is also generated by the function \texttt{randn}). Finally, observe that the DTLSS $\Sigma$ is minimal (note that the definition of minimality for linear switched systems are made by considering \emph{all} possible switching sequences in $Q^*$ \cite{petrbako}), however for the switching sequences restricted by the NDFA $\mathcal{A}$, it turns out $3$ states are disposable. In fact, this is the main idea of the paper.

%\begin{figure}[H]
%\centering
%\includegraphics[width=0.5\textwidth]{./fig/cdc_switching_sequence.eps}
%\caption{The switching sequence of length $11$ generated by $\mathcal{A}$ used in the simulation of both examples.}
%\label{fig:switching_sequence}
%\end{figure}

\begin{figure}[H]
\centering
\includegraphics[width=0.5\textwidth]{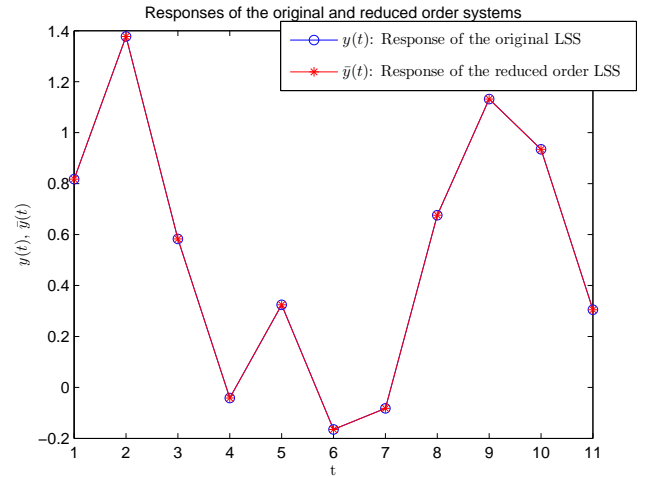}
\caption{Example $1$: The responses of the original DTLSS $\Sigma$ of order $7$ and the reduced order DTLSS $\bar{\Sigma}$ of order $4$ acquired by Algorithm~\ref{alg3} for the switching sequence in \eqref{eq:ex_switchsig}.}
\label{fig:output_ex1}
\end{figure}

One more example will be presented to illustrate the case when a representation for the $L$-unobservability space is constructed. The NDFA $\mathcal{A}$ accepting the language of allowed switching sequences is the same one used in the first example; whereas, this time the parameters of the SISO DTLSS $\Sigma$ of order $n=7$ with $Q=\{1,2,3\}$ are in the form:

\footnotesize
\begin{align*}
& \texttt{A}_{\texttt{1}}=\left[[\texttt{0}, \mbox{ } \texttt{1}, \mbox{ } \texttt{0}, \mbox{ } \texttt{0}, \mbox{ } \texttt{0}, \mbox{ } \texttt{0}, \mbox{ } \texttt{0}]; \mbox{ } \texttt{randn(6,7)}\right] \\
& \texttt{A}_{\texttt{2}}= \left[\texttt{randn(2,7)}; \mbox{ } [\texttt{0}, \mbox{ } \texttt{0}, \mbox{ } \texttt{0}, \mbox{ } \texttt{0}, \mbox{ } \texttt{0}, \mbox{ } \texttt{0}, \mbox{ } \texttt{0}]; \mbox{ }  \texttt{randn(4,7)}\right] \\
& \texttt{A}_{\texttt{3}}= \left[\texttt{randn(1,7)}; \mbox{ } [\texttt{0}, \mbox{ } \texttt{0}, \mbox{ } \texttt{1}, \mbox{ } \texttt{0}, \mbox{ } \texttt{0}, \mbox{ } \texttt{0}, \mbox{ } \texttt{0}]; \mbox{ }  \texttt{randn(5,7)} \right] \\
& \texttt{B}_{\texttt{1}}= \texttt{randn(7,1)} \\
& \texttt{B}_{\texttt{2}}= \texttt{randn(7,1)} \\
& \texttt{B}_{\texttt{3}}= \texttt{randn(7,1)} \\
& \texttt{C}_{\texttt{1}}= \texttt{randn(1,7)} \\
& \texttt{C}_{\texttt{2}}= \left[\texttt{1}, \mbox{ } \texttt{0}, \mbox{ } \texttt{0}, \mbox{ } \texttt{0}, \mbox{ } \texttt{0}, \mbox{ } \texttt{0}, \mbox{ } \texttt{0} \right] \\
& \texttt{C}_{\texttt{3}}= \texttt{randn(1,7)}
\end{align*}
\normalsize
Applying Algorithm~\ref{alg3} to this DTLSS, yields a reduced order system $\bar{\Sigma}$ of order $r=3$, whose output values are the same as the original system $\Sigma$ for the instances when the NDFA reaches a final state. Note that this corresponds precisely to an $L$-realization in the sense of Definition~\ref{def:lreal} (the last outputs of $\Sigma$ and $\bar{\Sigma}$ are the same for all the switching sequences generated by the governing NDFA, i.e., for all $\sigma \in L(\mathcal{A})$). In this example, the algorithm makes use of Lemma~\ref{theo:mert2} and constructs the $W$ matrix. The matrix $W \in \mathbb{R}^{r \times n}$ computed is $W= \left[ \begin{array}{c|c} I_3 & \mathbf{0} \end{array} \right]$.
%\begin{equation*}
%W= \left[ \begin{array}{c|c}
%I_3 & \mathbf{0}
%\end{array} \right].
%\end{equation*}

In this example, note that the resulting DTLSS is merely an $L$-realization of $\Sigma$ and nothing more as stated in Remark~\ref{rem:lemma4}, i.e., its output coincides with the output of $\Sigma$ for only the instances corresponding to the final states of the NDFA. This fact is visible from Figure~\ref{fig:output_ex2}, where it can be seen that output corresponding to the final state $s_f$ of the NDFA coincides for $\Sigma$ and $\bar{\Sigma}$ (Observe that for all switching sequences generated by $\mathcal{A}$ ending with the label $2$, the output values of $\Sigma$ and $\bar{\Sigma}$ are the same). Again, the input sequence of length $11$ used in the simulation is generated by the function \texttt{randn}. Finally, note that the DTLSS $\Sigma$ is again minimal whereas for the switching sequences restricted by the NDFA $\mathcal{A}$, it turns out $4$ states are disposable in this case.

\begin{figure}[H]
\centering
\includegraphics[width=0.5\textwidth]{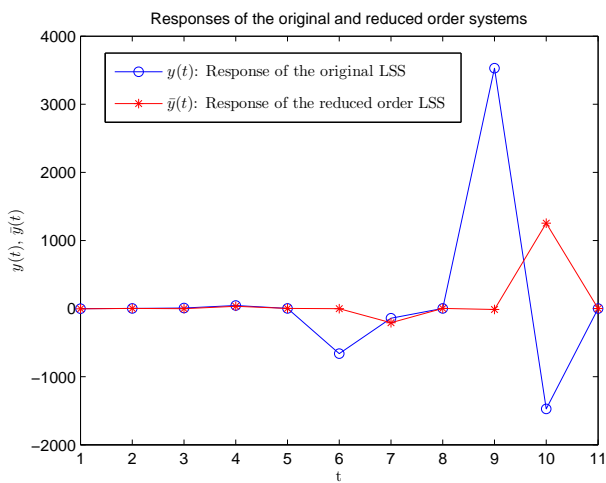}
\caption{Example $2$: The responses of the original DTLSS $\Sigma$ of order $7$ and the reduced order DTLSS $\bar{\Sigma}$ of order $3$ acquired by Algorithm~\ref{alg3} for the switching sequence in \eqref{eq:ex_switchsig}. Note that only the second and fifth element of the output sequence is equal for $\Sigma$ and $\bar{\Sigma}$ in first five elements, others look the same as a result of the scaling.}
\label{fig:output_ex2}
\end{figure}

\section{CONCLUSIONS}

A model reduction method for discrete time linear switched systems whose discrete dynamics are restricted by switching sequences comprising a regular language is presented. The method is essentially a moment matching type of model reduction method, which focuses on matching the Markov parameters of a DTLSS related to the specific switching sequences generated by a nondeterministic finite state automaton. %The method relies on constructing the partial reachability or unobservability spaces for the DTLSS related to the switching sequence language and the resulting reduced order system reproduces the same output as the original one for the last time instances of an allowed switching sequence. 
Possible future research directions include expanding the method for continuous time case, and approximating the input/output behavior of the original system rather than exactly matching it, and formulating the presented algorithms in terms of bisimulation instead of input-output equivalence. 
%%%%%%%%%%%%%%%%%%%%%%%%%%%%%%%%%%%%%%%%%%%%%%%%%%%%%%%%%%%%%%%%%%%%%%%%%%%%%%%%

%%%%%%%%%%%%%%%%%%%%%%%%%%%%%%%%%%%%%%%%%%%%%%%%%%%%%%%%%%%%%%%%%%%%%%%%%%%%%%%%

%%%%%%%%%%%%%%%%%%%%%%%%%%%%%%%%%%%%%%%%%%%%%%%%%%%%%%%%%%%%%%%%%%%%%%%%%%%%%%%%
%\section*{APPENDIX}
%\section*{ACKNOWLEDGMENT}
%
%The preferred spelling of the word ÒacknowledgmentÓ in America is without an ÒeÓ after the ÒgÓ. Avoid the stilted expression, ÒOne of us (R. B. G.) thanks . . .Ó  Instead, try ÒR. B. G. thanksÓ. Put sponsor acknowledgments in the unnumbered footnote on the first page.
%
%
%
%%%%%%%%%%%%%%%%%%%%%%%%%%%%%%%%%%%%%%%%%%%%%%%%%%%%%%%%%%%%%%%%%%%%%%%%%%%%%%%%%
%

%%\section*{APPENDIX}

%%\begin{Lemma} \label{lem:image}
%%$i$. Let $V \in \mathbb{R}^{n \times r}$ be any full column rank matrix and $V^{-1}$ be any left inverse of $V$. Then, $VV^{-1}u=u$ for any $u \in \IM V$.
%%
%%$ii$. Let $W \in \mathbb{R}^{r \times n}$ be any full row rank matrix and $W^{-1}$ be any right inverse of $W$. Then, $u^TW^{-1}W=u^T$ for any $u \in \IM W^T$.
%%\end{Lemma}

%%\begin{proof}
%%$i$. For each $u\in\IM V$ there exists a $x_1 \in \mathbb{R}^r$ such that $u=Vx_1$. Thus
%%\begin{equation} \label{image_equation}
%%VV^{-1}u=VV^{-1}Vx_1=Vx_1=u.
%%\end{equation}
%%$ii$. Proof of this part follows by duality from part $i$.
%%\end{proof}

%\section*{ACKNOWLEDGMENT}

\bibliographystyle{plain}
\bibliography{./root}

\begin{thebibliography}{10}

\bibitem{bastug_acc}
M.~Bastug, M.~Petreczky, R.~Wisniewski, and J.~Leth.
\newblock Model reduction by moment matching for linear switched systems.
\newblock {\em Proceedings of the American Control Conference}, to be
  published.

\bibitem{BempObs1}
A.~Bemporad, G.~Ferrari-Trecate, and M.~Morari.
\newblock Observability and controllability of piecewise affine and hybrid
  systems.
\newblock {\em IEEE Transactions on Automatic Control}, 45(10):1864--1876,
  2000.

\bibitem{French1}
A.~Birouche, J~Guilet, B.~Mourillon, and M~Basset.
\newblock Gramian based approach to model order-reduction for discrete-time
  switched linear systems.
\newblock In {\em Proc. Mediterranean Conference on Control and Automation},
  2010.

\bibitem{Chahlaoui}
Y.~Chahlaoui.
\newblock Model reduction of hybrid switched systems.
\newblock In {\em Proceeding of the 4th Conference on Trends in Applied
  Mathematics in Tunisia, Algeria and Morocco, May 4-8, Kenitra, Morocco},
  2009.

\bibitem{CalinBelta1}
Mircea Lazar Calin~Belta Ebru Aydin~Gol, Xu Chu~Ding.
\newblock Finite bisimulations for switched linear systems.
\newblock In {\em IEEE Conference on Decision and Control (CDC) 2012, Maui,
  Hawaii}, 2012.

\bibitem{PHaver}
Goran Frehse.
\newblock Phaver: algorithmic verification of hybrid systems past hytech.
\newblock {\em International Journal on Software Tools for Technology
  Transfer}, 10(3):263--279, 2008.

\bibitem{Lam1}
H.~Gao, J.~Lam, and C.~Wang.
\newblock Model simplification for switched hybrid systems.
\newblock {\em Systems \& Control Letters}, 55:1015--1021, 2006.

\bibitem{Habets1}
C.G.J.M. Habets and J.~H. van Schuppen.
\newblock Reduction of affine systems on polytopes.
\newblock In {\em International Symposium on Mathematical Theory of Networks
  and Systems}, 2002.

\bibitem{Kotsalis2}
G.~Kotsalis, A.~Megretski, and M.~A. Dahleh.
\newblock Balanced truncation for a class of stochastic jump linear systems and
  model reduction of hidden {M}arkov models.
\newblock {\em IEEE Transactions on Automatic Control}, 53(11), 2008.

\bibitem{Kotsalis1}
G.~Kotsalis and A.~Rantzer.
\newblock Balanced truncation for discrete-time {M}arkov jump linear systems.
\newblock {\em IEEE Transactions on Automatic Control}, 55(11), 2010.

\bibitem{liberzon2003}
D.~Liberzon.
\newblock {\em Switching in Systems and Control}.
\newblock Birkh{\"a}user, Boston, MA, 2003.

\bibitem{Mazzi1}
E.~Mazzi, A.S. Vincentelli, A.~Balluchi, and A.~Bicchi.
\newblock Hybrid system model reduction.
\newblock In {\em IEEE International conference on Decision and Control}, 2008.

\bibitem{6209392}
N.~Monshizadeh, H.~Trentelman, and M.~Camlibel.
\newblock A simultaneous balanced truncation approach to model reduction of
  switched linear systems.
\newblock {\em Automatic Control, IEEE Transactions on}, PP(99):1, 2012.

\bibitem{MP:Phd}
M.~Petreczky.
\newblock {\em Realization Theory of Hybrid Systems}.
\newblock PhD thesis, Vrije Universiteit, Amsterdam, 2006.

\bibitem{MP:BigArticlePartI}
M.~Petreczky.
\newblock Realization theory for linear and bilinear switched systems: formal
  power series approach - part i: realization theory of linear switched
  systems.
\newblock {\em ESAIM Control, Optimization and Caluculus of Variations},
  17:410--445, 2011.

\bibitem{petrbako}
M.~Petreczky, L.~Bako, and J.~H. van Schuppen.
\newblock Realization theory of discrete-time linear switched systems.
\newblock {\em Automatica}, 49:3337–--3344, November 2013.

\bibitem{petreczky2013}
M.~Petreczky, R.~Wisniewski, and J.~Leth.
\newblock Balanced truncation for linear switched systems.
\newblock {\em Nonlinear Analysis: Hybrid Systems}, 10:4--20, November 2013.

\bibitem{Shaker1}
H.R. Shaker and R.~Wisniewski.
\newblock Generalized gramian framework for model/controller order reduction of
  switched systems.
\newblock {\em International Journal of Systems Science, in press}, 2011.

\bibitem{Sun:Book}
Z.~Sun and S.~S. Ge.
\newblock {\em Switched linear systems : control and design}.
\newblock Springer, London, 2005.

\bibitem{TabuadaBook}
P.~Tabuada.
\newblock {\em Verification and Control of Hybrid Systems: A Symbolic
  Approach}.
\newblock Springer-Verlag, 2009.

\bibitem{HybSys}
Arjan van~der Schaft and Hans Schumacher.
\newblock {\em An Introduction to Hybrid Dynamical Systems}.
\newblock Springer-Verlag London, 2000.

\bibitem{YordanovBelta2}
B.~Yordanov and C.~Belta.
\newblock Formal analysis of discrete-time piecewise affine systems.
\newblock {\em Automatic Control, IEEE Transactions on}, 55(12):2834--2840, Dec
  2010.

\bibitem{China2}
L.~Zhang, E.~Boukas, and P.~Shi.
\newblock Mu-dependent model reduction for uncertain discrete-time switched
  linear systems with average dwell time.
\newblock {\em International Journal of Control}, 82(2):378-- 388, 2009.

\bibitem{China3}
L.~Zhang and P.~Shi.
\newblock Model reduction for switched lpv systems with average dwell time.
\newblock {\em IEEE Transactions on Automatic Control}, 53:2443--2448, 2008.

\bibitem{Zhang20082944}
L.~Zhang, P.~Shi, E.~Boukas, and C.~Wang.
\newblock Model reduction for uncertain switched linear discrete-time systems.
\newblock {\em Automatica}, 44(11):2944 -- 2949, 2008.

\end{thebibliography}

\addtolength{\textheight}{-12cm}   

\end{document}